\documentclass[12pt]{article}
\usepackage{amsmath}
\usepackage{amsfonts}
\usepackage{amssymb}
\usepackage{amsthm}
\usepackage{amstext}
\usepackage{graphicx}

\usepackage{color}

\newtheorem{teorema}{Theorem}[section]
\newtheorem{lema}[teorema]{Lemma}
\newtheorem{corolario}[teorema]{Corollary}
\newtheorem{proposicion}[teorema]{Proposition}
\theoremstyle{definition}
\newtheorem{remark}[teorema]{Remark}
\newtheorem{ejemplo}[teorema]{Example}

\def\A{\mathcal{A}}
\def\Har{\mbox{Har}\,}
\def\div{\mbox{div}\,}
\def\grad{\mbox{grad}\,}
\def\curl{\mbox{curl}\,}
\def\H{\mathbb{H}}
\def\C{\mathbb{C}}
\def\M{\mathfrak{M}}
\def\R{\mathbb{R}}
\def\V{\mathbf{V}}
\def\Sc{\mbox{Sc}\,}
\def\Vec{\mbox{Vec}\,}
\def\vecx{{\vec x}}
\def\vecy{{\vec y}}
\def\iff{\Leftrightarrow}
\def\Sol{\mbox{Sol}\,}
\def\Irr{\mbox{Irr}\,}
\def\Ti{T_{0,\Omega}}
\def\Tii{{\overrightarrow T}_{\!\!1,\Omega}}
\def\Tiii{\overrightarrow T_{\!\!2,\Omega}}

\begin{document}

\begin{center}
 {\Large General Solution of the Inhomogeneous\\
  Div-Curl System and Consequences}\\
 
\bigskip
Briceyda B. Delgado

R. Michael Porter

\bigskip
Departamento de Matem\'aticas, \\
Cinvestav-Quer\'etaro, Mexico

 \today

\end{center} 
 
\begin{abstract}
  We consider the inhomogeneous div-curl system (i.e.\ to find a
  vector field with prescribed div and curl) in a bounded star-shaped
  domain in $\R^3$. An explicit general solution is given in terms of
  classical integral operators, completing previously known results
  obtained under restrictive conditions. This solution allows us to solve
  questions related to the quaternionic main Vekua equation
  $DW=(Df/f)\overline W$ in $\R^3$, such as finding the vector part
  when the scalar part is known. In addition, using the general
  solution to the div-curl system and the known existence of the
  solution of the inhomogeneous conductivity equation, we prove the
  existence of solutions of the inhomogeneous double curl equation,
  and give an explicit solution for the case of static Maxwell's
  equations with only variable permeability.
\end{abstract}


\noindent \textbf{Classification:} 35Q60 (35Q61 30G20 30G35 32A26 35F35 35J15)

\section{Introduction}

We will give a complete solution to the reconstruction of a
vector field from its divergence and curl, i.e., the system
\begin{align}\label{div_rot_system}
 \div\vec w &= g_0,  \nonumber\\
 \curl\vec w &= \vec g ,
\end{align}
for appropriate assumptions on the scalar field $g_0$ and the vector
field $\vec g$ and their domain of definition in three-space.

This first order partial differential system governs, for example,
static electromagnetic fields.  In fact, Maxwell's equations consist
of two simultaneous div-curl systems which describe how electrical and
magnetic fields are generated by charges and currents together with
their variations. Basic references to the theory of the classical
Maxwell's equations are \cite{Boss1998,Jackson1999}. Chapters 3 and 4 of
\cite{Krav2003} develop a quaternionic treatment for different systems
of Maxwell's equations, and \cite[Chapter 2]{KravShap1996} does this for
electrodynamical models.

The div-curl system has been studied from very many points of view.
In \cite{BDS1982} an existence result for a solution of the related
Moisil-Teodorescu equation $Dw=g$ was proved, and the div-curl problem
consists of finding a purely vectorial solution. Explicit solutions
have been found under diverse restrictive conditions, either on the
data $g_0$ and $\vec g$ (beyond the evident requirement that $\vec g$
be solenoidal) or on the domain.  For example, in \cite[Section
4]{Bergman1969} a particular div-curl system with $g_0=0$ and
$\vec g=0$ is examined. On the other hand, for a solenoidal vector
field, that is, for $g_0=0$, the Biot-Savart vector fields
\cite{Fein2005,Gr1998} give a particular solution.  In \cite[Chapter
5]{Jiang1998} a numerical solution is given for the div-curl system
under certain boundary conditions, based on the Least-Squares finite
element method.  Another important solution of the div-curl system is
given in the reference book \cite[p.\ 166]{Korn1968} based on the
Helmholtz Decomposition Theorem, representing the solution as an
integral operator over all of three-space; this formula is not
applicable for, say, a bounded domain. A solution for star-shaped
domains, based on a radial integral operator, was recently provided by
Yu.\ M.\ Grigor'ev in \cite[Th.\ 3.2]{Gri2014}, valid when the
original data $g_0, \vec g$ in the system (\ref{div_rot_system}) are
harmonic scalar and vectorial functions, respectively. Somewhat
earlier, Colombo et.\ al.\ \cite{CLSSS2012} produced a right inverse
of curl under the condition that certain functions lie in
the kernel of one of the components of the Teodorescu operator. This
permits expressing the general solution for (\ref{div_rot_system})
under the assumption that a certain scalar field admits a
hyperconjugate harmonic function.

The present work may be considered as a completion of the analysis in
\cite{CLSSS2012}.  We will show that in fact that the required
hyperconjugate harmonic function exists whenever $\vec g$ is
solenoidal.  As in \cite{CLSSS2012} we rely heavily on the classical
Teodorescu operator, and for that reason we begin in Section
\ref{sec:quaternions} by presenting the terminology in the language of
quaternionic analysis, which we will mix freely with the notation of
the classical operators on vector fields.  All results obtained in
this paper are also valid for functions that take values in the the
algebra ${\H}({\C})$ of biquaternions (complex quaternions) but for
simplicity we will work with the real quaternions ${\H}$.  In Section
\ref{sec:teodorescu} we study the components of the Teodorescu
transform, which we apply in Section \ref{sec:div-curl} to solve the
div-curl problem in Theorem \ref{theorem_div_rot} by first
constructing an explicit inverse to the curl, a result which is of
independent interest. With this inverse we solve the homogeneous
div-curl system (in which $g_0$ vanishes), and then follow
\cite{CLSSS2012} to show how to apply a correction to obtain the
solution for the inhomogeneous system.  In the remaining sections we
apply this solution to several related problems, including some
Dirichlet-type problems, the conductivity equation, the main Vekua
equation, and the double curl-type equation, the latter of which is then
used in a fundamental way for solving the static Maxwell's equations
with variable permeability (system (\ref{system_4}) below).
To make the work self-contained and to highlight the beauty of
the interrelationships involved, we have included proofs of many
facts which can be found elsewhere.

The authors are pleased to express their gratitude to V.~V.~Kravchenko for
his extremely valuable suggestions and encouragement, without which
this work would have been impossible.

\section{Quaternionic analysis focused on
  $\R^3$\label{sec:quaternions}}

The multiplicative unit of the non-commutative algebra $\H$ of
quaternions is denoted $e_0=1$, while the nonscalar units are
$e_1,e_2,e_3$.  We generally consider an element
$x=x_0+\sum_{i=1}^3{e_i x_i}\in\H$ ($x_i\in\R$) to be decomposed as
$x=\Sc x + \Vec x$, where $\Sc x=x_0$; thus we have a direct sum
$\H=\Sc\H\otimes\Vec\H$. From now on, we will freely identify $\Sc\H$,
$\Vec\H$ with the real numbers $\R$ and Euclidean space $\R^3$
respectively.  Thus we have function spaces such as
$C^r(\Omega,\R),C^r(\Omega,\R^3)\subseteq C^r(\Omega,\H)$ for a domain
$\Omega$, which will always be contained in $\R^3$.  We largely follow
the notation in \cite{GuHaSpr2008}.

\subsection{Monogenic functions}

From now on $\vecx \in\R^3$.  The Moisil-Teodorescu differential operator
$D$ (also known as the Cauchy-Riemann or occasionally the Dirac
operator) is defined by
\begin{align}\label{operador_Dirac}
  D = e_1\frac{\partial}{\partial x_1} + e_2\frac{\partial}{\partial x_2} +
      e_3\frac{\partial}{\partial x_3}.
\end{align}
As $D$ may be applied from both the left and right sides, we write out
for clarity that for with a scalar function $w_0(\vecx)$, and a vectorial
function $\vec w(\vecx)$, 
\begin{align}\label{eq:Dvec}
   Dw_0 &= w_0D = \grad w_0,  \nonumber\\
   D\vec w &= -\div\vec w+\curl\vec w,
   \quad  \vec wD = -\div\vec w-\curl\vec w,
\end{align}
expressing $D$ in terms of the gradient $\nabla$, the divergence
$\nabla\cdot$ and the curl (or rotational) $\nabla\times$. Thus for
$w=w_0+\vec w$ the left and right operators are
\begin{align} 
  Dw &= Dw_0 + D\vec w = -\div\vec w + \grad w_0 + \curl\vec w, \nonumber\\
  wD &= Dw_0 - D\vec w = -\div\vec w + \grad w_0 - \curl\vec w.  \label{formula_Dirac}
\end{align}

The following \cite{GuSpr1990,GuSpr1997} is a generalization of the Leibniz rule:
\begin{align}\label{regla_Leibniz}
   D[vw]=D[v] w+\overline{v} D[w]+2 (\Sc(vD))[w],
\end{align}
where we write  
\[ (\Sc(vD))[w] =-\sum_{i=1}^3{v_i \partial_i w}
\] and $\overline{x}=\Sc x-\Vec x$ denotes quaternionic conjugation.
When $\Vec v=0$, this simplifies to
\[
D[vw]=D[v]w+vD[w].
\]
 
Let $\Omega\subseteq\R^3$ be an open subset.  A function
$w\in C^1(\Omega,\H)$ is called \textit{left-monogenic} (respectively
\textit{right-monogenic}) in $\Omega$ when $Dw=0$ (respectively
$wD=0$) and we write $\M(\Omega)=\mathfrak{M}(\Omega,\H)$ and
$\M^r(\Omega)=\M^r(\Omega,\H)$ for the spaces of left-monogenic and
right-monogenic functions. The unqualified term ``monogenic'' will
refer to left-monogenic functions; the term ``hyperholomorphic'' is
also commonly used.  By (\ref{formula_Dirac}),
\begin{equation}\label{monogenicas_izquierda}
 w\in\M(\Omega) \iff \left\{
    \parbox{.28\textwidth}{\vspace{-3ex}\begin{eqnarray*}
        \div\vec w \!\!&=&\!\! 0,\\ 
   \grad w_0  \!\!&=&\!\! -\curl\vec w .
   \end{eqnarray*} \vspace{-4ex}}
\right. 
\end{equation}

One sometimes says that $w_0, \vec w$ form a hyperconjugate
pair. We write
$\Har(\Omega,A)=\{w\colon\Omega\to A,\ \Delta w=0\}$, where $A=\R$,
$\R^3$ or $\H$, for the corresponding sets of harmonic functions
defined in $A$.  Since the Laplacian of a scalar function is obtained
by $\Delta w_0=-D^2w_0$, left and right monogenic functions are
harmonic.
  When both $Dw=0$ and $wD=0$, $w$ is called a \textit{monogenic
  constant}.  By (\ref{formula_Dirac}), $w$ is a monogenic constant if
and only if $w_0$ is constant and $\vec w$ satisfies $\div\vec w=0$
and $\curl\vec w=0$.  If $w\in\M(\Omega)$ with $\Sc w=0$ or
$\Vec w=0$, then $w$ is a monogenic constant. From this it can be seen
that the space
$\M^c(\Omega)=\M^c(\Omega,\H)=\M(\Omega)\cap\M^r(\Omega)$ of monogenic
constants in $\Omega$ can be decomposed as
\[ \M^c(\Omega) = \R\oplus \overrightarrow{\mathfrak{M}}(\Omega),
\]
where
\begin{align}\label{Si_vector_field}
  \overrightarrow{\M}(\Omega) =
 \Sol(\Omega,\R^3) \cap \Irr(\Omega,\R^3), 
\end{align}
with 
\begin{align*}
  \Sol(\Omega,\R^3) &= \{\vec w\colon\ \div\vec w=0 \text{ in } \Omega\}
    \subseteq C^1(\Omega,\R^3),\\
  \Irr(\Omega,\R^3) &= \{\vec w\colon\ \curl\vec w=0 \text{ in } \Omega\}
    \subseteq C^1(\Omega,\R^3).
\end{align*}
Elements of $\Sol(\Omega,\R^3)$ are called solenoidal (or
incompressible, or divergence free) fields, while elements of
$\Irr(\Omega,\R^3)$ are called irrotational vector fields.  Elements
of $\overrightarrow{\M}(\Omega)$ are called SI-vector fields and are
studied in \cite{GLS2010,Shapiro1997,Stein1971}.  Locally they are
gradients of real valued harmonic functions.  

\subsection{Standard integral operators}

The operators and results in this subsection are all well known.  Let
$\vec g=g_1e_1+g_2e_2+g_3e_3$ be a vector field such that
$\curl\vec g=0$.  Define \cite{Krav2009,KravTremb2011}
\begin{align}\label{antigradient}
  \A[\vec g](x_1,x_2,x_3) &= 
    \int_{a_1}^{x_1}{g_1(t,a_2,a_3) \,dt} +
    \int_{a_2}^{x_2}{g_2(x_1,t,a_3) \, dt} \nonumber\\
   & \ \ + \int_{a_3}^{x_3}{g_3(x_1,x_2,t) \,dt}.
\end{align}
Then the scalar function $\psi=\A[\vec g]$ is a potential (or
antigradient) for $\vec g$; i.e.\ $\grad\psi=\vec g$. Since potentials are
defined up to an arbitrary additive constant, this local definition 
can be extended to give
$\A\colon\Irr(\Omega,\R^3)\to C^2(\Omega,\R)$ whenever $\Omega$ is
simply connected. 

It is also well known \cite{Sud1979} that every real-valued harmonic function is
the scalar part of a monogenic function; conversely, the condition for
completing a vector part to a hyperconjugate pair is for $\vec w$ to
be harmonic and solenoidal:
\begin{proposicion} \label{prop:completingvecw} Let
  $\vec w\in\Har(\Omega,\R^3)$ where $\Omega$ is simply connected. A
  necessary and sufficient condition for there to exist
  $w\in\M(\Omega)$ such that $\Vec w=\vec w$ is that $\div\vec w=0$.
\end{proposicion}
 
\begin{proof}
  The necessity is given by (\ref{monogenicas_izquierda}). To prove
  the sufficiency, let $\vec w$ be solenoidal. Then
  $\curl\curl\vec w = \grad\div\vec w-\Delta\vec w = 0$, where
  $\Delta\vec w$ is the Laplacian applied to each component of the
  vector field.  Thus we can define $w_0=-\A[\curl\vec w]$ so that
  $\curl \vec w=-\grad w_0$ as required by
  (\ref{monogenicas_izquierda}).
\end{proof}

The radial moment operator, applicable to $\R^n$-valued functions
in general, is
\begin{align}  \label{eq:iradial}
  I^\alpha[w](\vecx) = \int_0^1 t^\alpha w(t\vecx)\,dt
\end{align}
in star-shaped domains, where usually $\alpha>-1$. Via relations such
as $\partial w_0(t\vecx)/\partial t=\vecx \cdot \grad w_0(t\vecx)$ one
verifies the following \cite{Gri2014}.

\begin{lema} \label{lemm:iradial} $\div I^\alpha = I^{\alpha+1}\div$;
  $\grad I^\alpha = I^{\alpha+1}\grad$;
  $\curl I^\alpha = I^{\alpha+1}\curl$;
  $\Delta I^\alpha = I^{\alpha+2}\Delta$;
  $\vecx \cdot I^{\alpha+1}[\vec w]=I^\alpha[ \vecx \cdot\vec w]$;
  $\vecx \times I^{\alpha+1}[\vec w]=I^\alpha[ \vecx \times\vec w]$;
  $I^\alpha[(\vecx \cdot \grad) w]=(\vecx \cdot \grad
	)I^\alpha[w]$ and
\[ I^\alpha[(\vecx \cdot\grad) w] = w -(\alpha+1)I^\alpha[w]. \]
\end{lema}
A further property we will need is
$I^\alpha[\vecx \cdot \curl \vec w]=\vecx \cdot \curl I^{\alpha}[\vec w]$,
which yields
$I^\alpha[\vecx \cdot \Vec D\vec w]=\vecx \cdot \Vec DI^{\alpha}[\vec w]$.

The \textit{monogenic completion operator}
$\vec S_{\Omega}\colon\Har(\Omega,\R)\to\Har(\Omega,\R^3)$ is the composition
\begin{align}\label{def:operator_completation_monogenic}
\vec S_{\Omega} = I^0[\Vec \vecx D]
\end{align} 
for star-shaped open sets $\Omega$ with respect to the origin.
(Recall that $Du$ is vectorial for scalar valued $u$; we have written
$\vecx D$ for the operator $(\vecx D)[u](\vecx)=\vecx Du(\vecx)$, which
involves a quaternionic multiplication.)  Explicitly this is
\[   \vec S_{\Omega}[w_0](\vecx) = \Vec\left(\int_0^1 t\vecx  D w_0(t\vecx)  \,dt\right) =
   \int_0^1 t \vecx \times  \nabla w_0(t\vecx)   \,dt, \quad \vecx \in \Omega.
\]
  When $\Omega$ is star-shaped with respect to some other
point, the definition of $\vec S_{\Omega}$ is adjusted by shifting the
values of $\vecx $ accordingly.  Versions of $\vec S_\Omega$ in $\R^n$ can
be found in greater generality in \cite{BDSH2006} and \cite[Sect.\
2.1.5]{GuSpr1997}; we give the proof of the following here for
completeness, modifying slightly the argument which was given in
\cite{Sud1979} for functions in domains in $\H$.

\begin{proposicion}\label{completacion_monogenicas} 
  Let $\Omega\subseteq\R^3$ be a star-shaped open set. The operator
  $\vec S_{\Omega}$ sends $\Har(\Omega,\R)$ to
  $\Har(\Omega,\R^3)$. For every real-valued harmonic function
  $w_0\in\Har(\Omega,\R)$,
\[ w_0+ \vec S_{\Omega}[w_0] \in \M(\Omega) .\]
Thus there is a monogenic function $w$ such that $\Sc w=w_0$.
\end{proposicion}

\begin{proof} 
 Let $w(\vecx) = w_0(\vecx) + \vec S_{\Omega}[w_0](\vecx)$.  Then since
$\Sc \vecx D[w_0]=-\vecx \cdot\grad w_0$, by Lemma \ref{lemm:iradial} we have
$ \Sc I^0[-\vecx D[w_0]] = w_0 -I^0[w_0]$, 
so (\ref{def:operator_completation_monogenic}) says
\begin{align*}
 w &= -I^0[D[w_0]\vecx] + I^0[w_0]\\ 
   &=\int_0^1 -tDw_0(t\vecx)\vecx  \,dt + \int_0^1 w_0(t\vecx) \,dt, 
\end{align*}
when $D[w_0]\vecx$ means the quaternionic multiplication $D[w_0](\vecx)\vecx$.
We apply $D$ and change the order of integration and derivation since
$w_0$ and $Dw_0$ have continuous partial derivatives in $\Omega$:
\begin{align}\label{eq.regla_cadena}
 (Dw)(\vecx) = \int_0^1 -t D_{\vecx}(D_{\vecx}w_0(t\vecx)\vecx) \,dt + \int_0^1 D_{\vecx}[w_0(t\vecx)] \,dt.
\end{align}
The subscript in $D_{\vecx}$ is the variable with respect to which we apply
the operator.  Using the Leibniz formula (\ref{regla_Leibniz}),
\begin{align*}
  D_{\vecx}\left(D_{\vecx}w_0(t\vecx) \vecx \right)
  &= -\Delta_{\vecx} (w_0(t\vecx))\vecx  + \overline{D_{\vecx} w_0(t\vecx)}D\vecx  -
    2\sum_{i=1}^3{\partial_i w_0(t\vecx)\partial_i {\vecx}}\\
  &= t\Delta_{\vecx} w_0(t\vecx)+3D_{\vecx}w_0(t\vecx)-2D_{\vec x}w_0(t\vecx)\\
  &= D_{\vecx}w_0(t\vecx)
\end{align*}
since $w_0$ is harmonic. 
Finally, since the second integrand in (\ref{eq.regla_cadena}) is
$tD_{\vecx}w_0(t\vecx)$, we conclude $Dw=0$ as required.
\end{proof}

Henceforth $\Omega\subseteq\R^3$ will always be a bounded domain.
  For a bounded function $w\in C(\Omega,\H)$, we can
define the volume integral
 \begin{align}\label{inverso_Laplaciano}
   L[w](\vecx) = -\int_{\Omega} \frac{w(\vec y )}{4\pi|\vec y -\vecx |} \,d\vecy ,
   \quad \vecx  \in \Omega,
\end{align}
while for $w\in C(\partial\Omega,\H)$ the \textit{single-layer
  potential} \cite[p.\ 38]{Colton1992} is the surface integral
\begin{align}\label{single_layer_operator}
    M[w](\vecx)=\int_{\partial\Omega} \frac{w(\vec y )}{4\pi|\vec y -\vecx |} \,ds_{\vecy} ,\quad
     \vecx \in \R^3\setminus \partial\Omega.
\end{align}
The \textit{Cauchy kernel} is the vector field
\begin{align*}
  E(\vecx) = \frac{\overline{\vecx }}{4\pi|\vecx |^3}, \quad \vecx  \in \R^3-\{0\}, 
\end{align*}
which is a monogenic constant.  For bounded $w\in C(\Omega,\H)$, the \textit{Teodorescu
  transform} of $w$ is defined by
\begin{align}\label{operador_Teodorescu}
T_{\Omega}[w](\vecx) =-\int_{\Omega} E(\vec y -\vecx) w(\vec y ) \,d\vecy, \quad \vecx \in\R^3.
\end{align}

\begin{proposicion}[{\cite[Prop.\ 2.4.2]{GuSpr1990}}]\label{prop:Teodorescu}  
      Let $w\in C(\Omega,\H)$ be bounded. Then
  $T_\Omega(w)\in C^1(\Omega,\H)$.  Further, the Teodorescu
 transform acts as the right inverse
  operator of $D$:
    \[ DT_\Omega[w]=w. \]  
\end{proposicion}

  \section{Components of the Teodorescu operator
\label{sec:teodorescu}}

As a preliminary to providing the general solution to the div-curl system
(\ref{div_rot_system}) in Theorem \ref{theorem_div_rot} below, we begin by
analyzing the elements which form the Teodorescu operator. The following
operators were introduced in \cite{CLSSS2012}:
\begin{align}\label{T1_T2_T3}
 \Ti[\vec w](\vecx) &=  \int_{\Omega} E(\vecy -\vecx) \cdot \vec w(\vecy ) \,d\vecy, \nonumber\\
 \Tii[w_0](\vecx)    &= -\int_{\Omega} w_0(\vecy ) E(\vecy -\vecx) \,d\vecy,  \nonumber\\
 \Tiii[\vec w](\vecx) &= -\int_{\Omega} E(\vecy -\vecx) \times \vec w(\vecy ) \,d\vecy,
\end{align}
where $\cdot$ denotes the scalar (or inner) product of vectors and
$\times$ denotes the cross product.  Note that $\Tii$ acts
on $\R$-valued functions, while $\Ti $, $\Tiii $ act
on $\R^3$-valued functions, and $\Ti $ produces scalar-valued
functions. 
Furthermore,
\begin{align}\label{decomposition_T}
  T_{\Omega}[w_0+\vec w] = \Ti[\vec w] + \Tii[w_0]+\Tiii[\vec w].
\end{align}
This is an expression of the quaternionic multiplication formula
$\vec a b= -\vec a\cdot\vec b + b_0\vec a + \vec a\times \vec b$.

The first statement of the following was noted in \cite[Prop.\ 3.8]{CLSSS2012}.

\begin{proposicion}\label{scalar_part_harmonic}
  Suppose that $\vec w$ is bounded.  
 (i) $\Ti[\vec w]\in \Har(\Omega,\R)$ if and only if
  $\vec w\in \Sol(\Omega,\R^3)$;
 (ii) $\Tiii[\vec w]\in\Har(\Omega,\R^3)$ if and only if
  $\vec w\in \Irr(\Omega, \R^3)$.
\end{proposicion}

\begin{proof}
  Using $\Delta=-D^2$ and the property $DT_{\Omega}=I$ of Proposition
  \ref{prop:Teodorescu} together with the decomposition of the
  operator $D$ given in (\ref{formula_Dirac}) it follows that
\begin{align*}
 \Delta T_{\Omega}[\vec w]=-D^2 T_{\Omega}[\vec w]  
  = -D \vec w 
  = \div\vec w - \curl\vec w.
\end{align*}
The scalar and vector parts are $\Delta \Ti[\vec w] =\div\vec w$ and
$\Delta \Tiii[\vec w] =-\curl\vec w$ by (\ref{decomposition_T}). 
\end{proof}

The operators $\Ti$, $\Tii$, $\Tiii $ can be expressed in terms of the
operator $L$ given in (\ref{inverso_Laplaciano}) acting on continuous
functions and fields.

\begin{proposicion} \label{operators_T1T2T3} For bounded   
  $w_0\in C(\Omega,\R)$, $\vec w\in C(\Omega,\R^3)$,
\begin{align*} 
  \Ti[\vec w] &= \nabla \cdot L[\vec w]\\
  \Tii[w_0] 
         &= - \nabla L[w_0], \\
  \Tiii[\vec w] 
     &= - \nabla \times L[\vec w].
\end{align*}
Consequently, $\Tii[w_0]\in\Irr(\Omega, \R^3)$ and
$\Tiii[\vec w]\in\Sol(\Omega,\R^3)$. 
\end{proposicion}
\begin{proof}
  The proof is a direct calculation,  using 
  $\nabla_{\vecx}(1/|\vecx -\vec y |)=-4\pi E(\vec y -\vecx)$ and the product rules of vector
  analysis \cite[Cor.\ 1.3.4]{GuSpr1990}. The conclusion regarding 
the images of $\Tii$, $\Tiii$ was already noted in \cite[Prop.\ 3.2]{CLSSS2012}.
\end{proof}

\begin{proposicion}\label{prop:inverselaplacian}
  The operator $L$ given by (\ref{inverso_Laplaciano}) is a right inverse of
  the Laplacian $\Delta$ on the space of bounded functions in
  $C^1({\Omega},\H)$.
\end{proposicion}

\begin{proof}
  Let $w=w_0+\vec w\in C^1({\Omega},\H)$. Using Proposition
  \ref{prop:Teodorescu}, the identity
  $\curl \curl \vec w=\grad \div \vec w-\Delta \vec w$ and the
  expressions given in Proposition \ref{operators_T1T2T3}, we
  have that
\begin{align*}
 w = DT_{\Omega}w&=D(\Ti[\vec w]+\Tii[w_0]+\Tiii[\vec w])\\
 &= \div(\grad L[w_0])+\grad(\div L[\vec w])-\curl(\curl L[\vec w])\\
 &= \Delta L[w]. 
\end{align*}
\end{proof}

Proposition \ref{prop:inverselaplacian} may also be proved using the
fact that $1/(4\pi|\vecx -\vecy|)$ is a fundamental solution for the Laplacian.
From the development given in the above proof we also have
\begin{align}
  \div \Tii = -\div\grad L = -\Delta L = -\text{identity}, 
\end{align}
which gives the following.

\begin{corolario}\cite[Prop.\ 3.1]{CLSSS2012} 
  The operator $-\Tii$ acting on bounded functions in $C(\Omega,\R)$
  is a right inverse for the divergence div.
\end{corolario}

Another useful relation is the following.

\begin{proposicion} \label{prop:DTidentity}
  The relation
  $\grad\, \Ti +\curl \Tiii= \mbox{identity}$     holds on  $C(\Omega,\R^3)$.
\end{proposicion}

\begin{proof}
 We have
  $T_{\Omega}[\vec g]=\Ti[\vec g]+\Tiii[\vec g]$, since $\vec g$ has
  no scalar part.  The statement is obtained from the  
  vectorial part of the relation of Proposition \ref{prop:Teodorescu}
  according to (\ref{formula_Dirac}).
\end{proof}

Most of what we have done will go through equally well in the context
of generalized derivatives. Thus
(\ref{inverso_Laplaciano})--(\ref{operador_Teodorescu}) make sense
when $w$ is integrable, and Proposition \ref{prop:Teodorescu} is shown
in \cite{GuSpr1990} to hold in $L^p(\Omega,\H)$. Thus Proposition
\ref{scalar_part_harmonic} extends to the situation in which
$\div\vec w=0$ or $\curl \vec w=0$ holds in the sense of
distributions, because by Weyl's Lemma \cite{Weyl1940}, \cite[Th.\
24.9]{Forster1981}, the weak solutions $\Ti[\vec w]$ and
$\Tiii[\vec w]$ of the Laplace equation are smooth solutions.

\section{Solution to the div-curl system\label{sec:div-curl}}

The first step in solving the div-curl system is to obtain an inverse
for the curl operator, an object which is of independent interest.  We
will use the monogenic completion operator $\vec S_{\Omega}$ of
(\ref{def:operator_completation_monogenic}) and the component
Teodorescu operators of (\ref{T1_T2_T3}).
Note that the vanishing divergences
$\div \vec S_{\Omega}[\Ti[\vec w]] = \div \Tiii[\vec w] = 0$
imply the a priori fact that 
\begin{align*}
\Tiii  - \vec S_{\Omega} \Ti \colon \Sol(\Omega,\R^3) \to \Sol(\Omega,\R^3).
\end{align*}

\begin{teorema}\label{theorem_inverse_rot}
  Let $\Omega\subseteq\R^3$ be a star-shaped open set. The operator
\begin{align}\label{inverse_rot}
   \Tiii  - \vec S_{\Omega} \Ti 
\end{align}
is a right inverse for the curl acting on the class of bounded
functions in $\Sol(\Omega,\R^3)$.
\end{teorema}

\begin{proof}
  Let $\vec g\in \Sol(\Omega,\R^3)$ and let
  $\vec w = \Tiii[\vec g]-\vec S_{\Omega}[v_0]$ where
  $v_0=\Ti[\vec g]$.  By Proposition \ref{scalar_part_harmonic},
  $v_0\in\Har(\Omega,\R)$, so by Proposition
  \ref{completacion_monogenicas}, $v_0+\vec S_{\Omega}[v_0]$ is a
  monogenic function whose equivalent system
  (\ref{monogenicas_izquierda}) is
\begin{align} \label{eq:divrotS}
 \div \vec S_{\Omega}[v_0] &= 0,   \nonumber\\
 \curl \vec S_{\Omega}[v_0] &= -\nabla v_0.
\end{align}
Combining these equations equations with Proposition
\ref{prop:DTidentity}, we have that
\begin{align*} \curl\vec w = 
  -\nabla v_0 + \vec g + \nabla v_0=
   \vec g.
\end{align*}
\end{proof}
 
In \cite{CLSSS2012} it was shown that $\Tiii$ acts as a right inverse
for curl \textit{for elements of the kernel of $\Ti$}. Indeed, let
$\Ti[\vec g]=0$. By Proposition \ref{scalar_part_harmonic}, the field
$\vec g$ is indeed solenoidal, and since by (\ref{eq:divrotS})
$\div\vec S_{\Omega}[\Ti[\vec g]]=0$, Theorem
\ref{theorem_inverse_rot} says that $\curl \Tiii[\vec g] =\vec g$.  It
was recognized in \cite{CLSSS2012} that to require $\Ti[\vec g]$ to
vanish would be too strong a condition; now we see that the precise
condition is for it to be harmonic.

\begin{corolario}\label{corollary_inverse_rot}
  Let $\Omega\subseteq\R^3$ be a star-shaped open set. Let
  $\vec g\in\Sol(\Omega,\R^3)$ be a bounded divergence free vector
  field. Then the general solution of the homogeneous system
\[ \div \vec w=0, \quad \curl \vec w =\vec g \]
in $\Omega$ has the form
\begin{align} \label{general_solution1}
 \vec w = \Tiii[\vec g] - \vec S_{\Omega}[\Ti[\vec g]] + \nabla h,
\end{align}
where $h\in\Har(\Omega,\R)$ is an arbitrary real-valued harmonic
function.
\end{corolario}

\begin{proof}
  First let $\vec w$ be given by (\ref{general_solution1}).  Then
  $\div\vec w=0$ by the observations preceding the statement
  of the Corollary. A difference $\vec v$ of two
  solutions of the div-curl system satisfies $\div v=0$, $\curl v=0$,
  i.e., $\vec v$ is a monogenic constant. Since $\Omega$ is
  star-shaped and therefore simply connected, $\vec v$ is the gradient
  of a harmonic function $h$.
\end{proof}

\begin{corolario}
  The operator
\begin{align} \label{inverse_rot_rot}
  -L - I^{-1}\left[\frac{|\vecx |^2}{2} \grad  \Ti\right] 
\end{align}
is a right inverse for the double curl operator $\curl\,\curl$ acting
on the class of bounded functions $\vec w\in\Sol(\Omega,\R^3)$.
\end{corolario}
(Note that $I^\alpha$ for the exponent $\alpha=-1$ can be applied
because of the factor $|\vecx |^2$ in the operand.)

\begin{proof}
  It is enough to show that the curl applied after the operator
  (\ref{inverse_rot_rot}) produces the right inverse of curl given in
  (\ref{inverse_rot}). But Proposition \ref{operators_T1T2T3} with
  (\ref{def:operator_completation_monogenic}) gives
\begin{align*}
 \curl I^{-1}\left[\frac{|\vecx |^2}{2} \grad  \Ti \right]
   &= I^0\left[\curl\left(\frac{|\vecx |^2}{2}\grad\Ti\right)\right] \\
   &= I^0\left[\grad_{\vecx} \frac{|\vecx |^2}{2} \times \grad\Ti \right]\\
   &= I^0\left[\vecx  \times \grad\Ti\right]\\
   &= \vec S_\Omega\,\Ti ,
\end{align*}
while by Lemma \ref{lemm:iradial},
\[ \curl L = -\Tiii. \]
Adding these equalities we have the result.
\end{proof}

Now we can proceed to solve the inhomogeneous div-curl system with data
$g_0,\vec g$. Assuming as always that $\vec g$ is solenoidal, the
function $v=T_{\Omega}[-g_0+\vec g]$ satisfies $Dv=-g_0+\vec g$
and therefore is a \textit{quaternionic} solution to
(\ref{div_rot_system}). We seek to construct a vector solution
$\vec w$ by subtracting a monogenic function whose scalar part is
precisely the scalar part of $v$. Thus the key consists in taking the
$\Ti $ component of $T_\Omega(\vec g)$, and using Proposition
\ref{completacion_monogenicas} to construct the monogenic conjugate of
the $\R$-valued function $\Ti[\vec g ]$. This is accomplished
in the following result.

\begin{teorema}\label{theorem_div_rot}
  Let $\Omega$ be a star-shaped open set. Let $g_0\in C(\Omega,\R)$
  and $\vec g\in\Sol(\Omega,\R^3)$ be bounded. The general solution of the
  inhomogeneous div-curl system (\ref{div_rot_system}) is given by
\begin{align}\label{general_solution2}
  \vec w = -\Tii[g_0] + \Tiii[\vec g] -
           \vec S_{\Omega}[\Ti[\vec g]] + \nabla h,
\end{align}
where $h\in\Har(\Omega,\R)$ is arbitrary.
\end{teorema}

\begin{proof}
  Since $\div\vec g=0$, Proposition \ref{scalar_part_harmonic} says
  that $\Ti[\vec g]$ is an $\R$-valued harmonic function, so
  Proposition \ref{completacion_monogenicas} permits us to complete it to
  the monogenic function
  $\Ti[\vec g ] + \vec S_{\Omega}[\Ti[\vec g]]$. By
  (\ref{decomposition_T}), the difference
\begin{align*}
    T_{\Omega}[-g_0+\vec g] - (\Ti[\vec g] + 
            \vec S_{\Omega}[\Ti[\vec g ]])  
         = -\Tii[g_0] + \Tiii[\vec g] -  \vec S_{\Omega}[\Ti[\vec g]]
\end{align*} 
is purely vectorial. By Proposition \ref{prop:Teodorescu} and the fact that
$D\nabla h=0$,
\begin{align*}
  D\vec w= D(T_{\Omega}[-g_0+\vec g]) = -g_0+\vec g .
\end{align*}
so (\ref{div_rot_system}) is satisfied because of (\ref{eq:Dvec}). As
in the proof of Corollary \ref{corollary_inverse_rot}, we obtain the
general solution adding to $\vec w$ the gradient of an arbitrary
harmonic function.
\end{proof}

 Consider the following Sobolev spaces,
\begin{align*}
  H^1(\Omega,A) &= \left\{u\in L^2(\Omega,A)\colon\
                    \grad u\in L^2(\Omega,A)\right\},\\
  H_0^1(\Omega,A) &= \overline{C_0^{\infty}(\Omega,A)}
                     \subseteq H^1(\Omega,A),\\
  H^{1/2}(\partial\Omega,A) &= \{\varphi\in L^2(\partial\Omega,A)\colon\
             \exists u\in H^1(\Omega,A) \text{ with }
             u|_{\partial\Omega}=\varphi \}.
\end{align*}
Again $A=\R,\R^3,\H$ while $C_0^\infty$ denotes smooth functions of
compact support; the space $H^{1/2}(\partial\Omega,A)$ consists of the
boundary values of functions in $H^1(\Omega,A)$. (Whenever
$\partial\Omega$ is mentioned we will assume it is sufficiently smooth
to apply the basic facts about Sobolev spaces.)  We write
$[\varphi] =u + H_0^1(\Omega,A)\subseteq H^1(\Omega,A)$ for the
equivalence class of functions with common boundary values $\varphi$.

Proposition \ref{scalar_part_harmonic} and the results of this section
are valid even in the distributional sense, because the properties of
the Teodorescu transform $T_{\Omega}$ (see \cite[Prop.\
2.4.2]{GuSpr1990}) as well as Proposition
\ref{completacion_monogenicas} are also valid in Sobolev spaces. Thus
in Theorem \ref{theorem_div_rot}, we may take $g_0\in L^2(\Omega,\R)$,
$\vec g\in L^2(\Omega,\R^3)$ and $\div \vec g=0$ in the weak sense,
\begin{align}
     \int_{\Omega} \vec g \cdot \nabla v \,dx    =0,
\end{align}
for all test functions $v\in H_0^1(\Omega)$. Then
(\ref{general_solution2}) is a weak solution of the div-curl system
(\ref{div_rot_system}).

\begin{ejemplo}\label{example} 
  Let $\Omega$ be the unit ball in $\R^3$ minus any ray emanating from
  the origin.  Take $g_0=0$ and $\vec g=\vecx/|\vecx|^3$ in the
  div-curl system (\ref{div_rot_system}). Since $\Omega$ is
  star-shaped with respect to any of its points $\vec a$ opposite to
  the ray, we shift the origin to $\vec a$ in the formula
  (\ref{def:operator_completation_monogenic}) for $\vec S_{\Omega}$
  (cf.\ \cite{Sud1979}) as follows:
\begin{align}
    \vec S_{\Omega}[w_0](\vecx)=\Vec\int_0^1{-t D w_0((1-t)\vec a+t\vecx) (\vecx-\vec a) dt}.
\end{align}
Since the removed ray is of zero measure, we may use the explicit
formula for the Teodorescu transform of $\vec g$ for the unit ball
given in \cite[p.\ 324, formula 28]{GuSpr1997}, namely
$\Ti[\vec g]=1-1/|\vecx|$ while $\Tii[g_0]=\Tiii[\vec g]=0$. Since
$\vec g$ is solenoidal, the div-curl solution of Theorem
\ref{theorem_div_rot} is
\begin{align}
 \vec w(\vecx)&=- \vec S_{\Omega}\left[-\frac{1}{|\vecx|}+1\right] \nonumber \\
    &=\frac{\vec a\times \vecx}{|\vecx|(\vec a\cdot \vecx-|\vec a||\vecx|)}. \label{eq:ejemplo}
\end{align}
Since $\div\vec w=0$ and $\curl\vec w=\vecx/|\vecx|^3$ are independent of $\vec a$,
the difference of two such solutions is an SI field, as would be expected.
\end{ejemplo}

\begin{remark}
  Suppose now that two given scalar and vectorial functions $g_0$ and
  $\vec g$ are harmonic. Under this additional hypothesis (and of
  course $\div\vec g=0$), a solution  
\begin{align} \label{eq:grig}
\vec w=-\vecx  \times I^1[\vec g]+\grad \left(\frac{|\vecx|^2}{4}I^{1/2}
    [g_0- \vecx \cdot I^2 [\curl\vec g]]\right)     
\end{align}
for the div-curl system (\ref{div_rot_system}) was given by Yu.\ M.\
Grigor'ev in \cite[Th.\ 3.2]{Gri2014}, where the integrals $I^\alpha$
were defined in (\ref{eq:iradial}).  We relate this to our solution
(\ref{general_solution2}).

By additivity we may consider $g_0$ and $\vec g$
independently. Suppose $g_0=0$, and let $\vec w$ be given by
(\ref{eq:grig}).  Substitute
$\vec g=\curl \Tiii[\vec g]+\grad \Ti[\vec g]$ (Proposition
\ref{prop:DTidentity}) to obtain
$\vec w = \vec f-\vec S_\Omega\Ti[\vec g]$ where
\[ \vec f(x) = -\vecx  \times I^1[\curl\Tiii\vec g] +
   \grad \left(\frac{|\vecx|^2}{4}I^{1/2}[-\vecx \cdot I^2 [\curl\curl\Tiii[\vec g]]]\right).
\]
Here we have used the fact that
$\curl\curl\Tiii[\vec g]= \curl[\vec g-\grad\Ti[\vec g]]=\curl\vec g$.

The function  $\vec v= \Tiii[\vec g]$ trivially satisfies
\begin{align*}
   \div \vec v &= 0,  \nonumber\\
   \curl \vec v&= \curl \Tiii[\vec g],
\end{align*}
while by (\ref{eq:grig}), $\vec f$ satisfies the same system.  Thus
$\vec f$ differs from $\Tiii[\vec g]$ by an SI field. Since $\Tii[g_0]=0$,
we have $\vec w = \Tiii[\vec g] - \vec S_\Omega\Ti[\vec g] +\nabla h$ for some
harmonic $h$, which agrees with our solution (\ref{general_solution2}).

 For
$\vec g=0$, we simply note that both (\ref{general_solution2}) and
(\ref{eq:grig}) are left inverses of div applied to $g_0$.

It is easily seen that $\vec g$ in Example \ref{example} is harmonic,
and that when one shifts appropriately the base point of integration
in (\ref{eq:grig}), the same solution (\ref{eq:ejemplo}) is
obtained.

\end{remark}

\subsection{Div-curl system with boundary data}
We will rewrite the right inverse for the curl (\ref{inverse_rot}) in
terms of boundary value integral operators. The following result tells
us that $\Ti$ is in some sense a boundary integral operator, as it can
be expressed in terms of the single-layer operator $M$ of
(\ref{single_layer_operator}) when the boundary values of $\vec w$ are
known. Here $\eta$ denotes the unit normal vector to the boundary
$\partial \Omega$, which from here on will be assumed to be smooth,
and $\Omega$ will be bounded.

\begin{proposicion}\label{T1_frontera_prop}
For every $\vec w\in \Sol(\overline{\Omega},\R^3)$,
\begin{align*} 
  \Ti[\vec w] = M[\vec w|_{\partial\Omega}\cdot \eta].
\end{align*}
\end{proposicion}
\begin{proof}
  Using that
\[ \nabla_{\vecx}\left(\frac{1}{4\pi|\vecx -\vecy |}\right) = 
  -\frac{\vecx -\vecy }{4\pi|\vecx -\vecy |^3} = 
     - \nabla_{\vecy}\left(\frac{1}{4\pi|\vecx -\vecy |}\right), 
\]
we find that
\begin{align}\label{T0_interior}
\nonumber  \Ti[\vec w](\vecx)
\nonumber  &= \int_{\Omega}{\frac{\vecx -\vecy }{4\pi|\vecx 
    -\vecy |^3}\cdot \vec w(\vecy ) \,d \vecy}\\
\nonumber  &= \int_{\Omega}{\nabla_{\vecy}\left(\frac{1}{4\pi|\vecx
     -\vec y |}\right)\cdot  \vec w(\vecy ) \,d \vecy}\\
           &= \int_{\Omega}{\nabla_{\vecy} \cdot
    \left(\frac{\vec w(\vecy )}{4\pi|\vecx -\vecy |}\right) \,d \vecy}.
\end{align}
since $\vec w\in\Sol(\Omega,\R^3)$. By the Divergence Theorem,
\begin{align*}
  \Ti[\vec w](\vecx) =
   \int_{\partial\Omega}{\frac{\vec w(\vecy )}{4\pi|\vecx -\vecy |}\cdot \eta(\vecy ) \,ds_{\vecy}} 
\end{align*}
as desired.
\end{proof}

Analogously to Proposition \ref{T1_frontera_prop}, the
operator $\Tiii $ can be recovered when we know
the boundary values of functions in $\Irr(\overline{\Omega},\R^3)$:

\begin{proposicion}\label{T3_frontera_prop}
  For every $\vec w\in \Irr(\overline{\Omega},\R^3)$,
\begin{align} \label{T3_frontera}
  \Tiii[\vec w] = - M[\vec w|_{\partial\Omega}\times \eta]  . 
\end{align}
\end{proposicion}

\begin{proof}
\begin{align*}
  \Tiii[\vec w](\vecx) &=
   \int_{\Omega} -\frac{\vecx -\vecy }{4\pi|\vecx -\vecy |^3}\times \vec w(\vecy ) \,d\vecy\\
  &= - \int_{\Omega} \nabla_{\vecy}\left(\frac{1}{4\pi|\vecx -\vecy |}\right)\times \vec w(\vecy ) \,d\vecy\\
  &= - \int_{\Omega} \nabla_{\vecy} \times \left(\frac{\vec w(\vecy )}{4\pi|\vecx -\vecy |}\right) \,d\vecy +
       \int_{\Omega} \frac{\nabla_y\times \vec w(\vecy )}{4\pi|\vecx -\vecy |} \,d\vecy.
\end{align*}
Apply $\vec w\in\Irr(\overline{\Omega},{\R}^3)$ and Stokes' theorem to obtain
the desired result.
\end{proof}

We will write $\overrightarrow{\M}(\partial\Omega)$ for the space of
boundary values of SI vector fields in $\overline{\Omega}$, which we
recall from (\ref{Si_vector_field}) are the purely vectorial monogenic
constants. Since SI vector fields are harmonic, the extension of
$\varphi\in \overrightarrow{\M}(\partial\Omega)$ to the interior in unique. We
rewrite the right inverse of \curl given in Theorem
\ref{theorem_inverse_rot} as a boundary integral operator, under the
condition that boundary data $\vec \varphi$ has an irrotational and
solenoidal extension:
 
\begin{proposicion}
  Let $\vec \varphi \in\overrightarrow{\M}(\partial\Omega)$ be the
  boundary values of the vector field $\vec g$ defined in
  $\Omega$. Define
\begin{align}\label{solution_boundary}
   \vec w = -M[\vec\varphi \times \eta] - \vec S_{\Omega}[M[\vec\varphi \cdot \eta]]
\end{align} 
where again $\eta$ is the outward normal. Then  
\begin{align*}
  \div \vec w=0,\quad \curl \vec w=\vec g .
\end{align*}
\end{proposicion}  
\begin{proof}
  By Propositions \ref{T1_frontera_prop} and \ref{T3_frontera_prop},
  $\Ti[\vec w]=M[\vec \varphi \cdot \eta]$ and
  $\Tiii[\vec g]=-M[\vec \varphi \times \eta]$, respectively. The
  statement now follows from the Corollary
  \ref{corollary_inverse_rot}.
\end{proof}
Let $\vec\varphi\in\overrightarrow{\M}(\partial\Omega)$.  Since
$\Delta \vec w=-\curl \vec g$, the solution (\ref{solution_boundary})
solves the following Dirichlet-type problem
\begin{align*}
\Delta \vec w&=0,\\
\curl \vec w|_{\partial\Omega}&=\vec \varphi.
\end{align*}
 
\section{Application to the main Vekua equation
  and conductivity equation on $\R^3$\label{sec:mainvekua}}

The Vekua equation, whose theory was introduced in
\cite{Bers1953,Vekua1959} for functions in $\R^2$, plays an important
role in the theory of pseudo-analytic functions (sometimes called
generalized analytic functions). We will study a special Vekua
equation, which in \cite{Krav2009} is called the main Vekua
equation. We are interested in the natural generalization of this
equation to the quaternionic case \cite[Ch.\ 16]{Krav2009}, which
possesses properties similar to those of the complex Vekua equation,
including an intimate relation with the conductivity equation. The conductivity
equation appears in many aspects of physics, and gives rise to inverse
problems with applications to fields such as tomography.  Here we
apply the results obtained on the div-curl system to study solutions
of these equations.

\subsection{The main Vekua equation and equivalent formulations}
 
The \textit{main Vekua equation} is
\begin{align}\label{ecuacion_Vekua}
   DW=\frac{Df}{f}\overline W,
\end{align}
with $D$ the Moisil-Teodorescu operator given in
(\ref{operador_Dirac}), and $f$ a nonvanishing smooth function.  We are
interested in solutions $W=W_0+\vec W\in C^1(\Omega,\H)$ for a domain
$\Omega\subseteq\R^3$.

The operator $D-(Df/f)C_{\H}$ corresponding to (\ref{ecuacion_Vekua}), and
similar expressions, appear in various factorizations. For example,
when $u$ is scalar,
\begin{align*}
  \nabla \cdot f^2 \nabla u =
   -f\left(D+M^{\frac{Df}{f}}\right)\left(D-\frac{Df}{f}C_{\H}\right)fu,\\
  \left(\Delta-\frac{\Delta f}{f}\right)u =
   \left(D+M^{\frac{Df}{f}}\right)\left(D-M^{\frac{Df}{f}}\right)u,
\end{align*}
where $C_{\H}$ is the quaternionic conjugate operator and $M^{a}$
denotes the operator of multiplication on the right by the function
$a$. 

The set  
\begin{align}\label{espacio_soluciones}
 \M_f(\Omega) = \left\{W\colon\  DW = 
   \frac{Df}{f}\overline{W}\right\}\subseteq C^1(\Omega,\H)
\end{align}
of solutions of the main Vekua equation (\ref{ecuacion_Vekua}) is a
nontrivial linear subspace over $\R$.  
In \cite[Chapter 16]{Krav2009} we find results that relate
solutions of the main Vekua equation to solutions of other
differential equations. In particular, (\ref{ecuacion_Vekua}) is
related to the $\R$-linear Beltrami equation, a fact which was
essential in the solution of the Calder\'{o}n problem in the complex
case \cite{AP2006}. A similar fact is the following.

\begin{lema}\cite[Th.\ 161]{Krav2009} \label{lem:system_1}
  $W\in\M_f(\Omega)$ if and only if the scalar part $W_0$ and the
  vector part $\vec W$ satisfy the homogeneous div-curl
  system
\begin{align}\label{system_1} 
  \div(f\vec W) &= 0, \nonumber\\
  \curl(f\vec W) &= -f^2\nabla \left(\frac{W_0}{f}\right).
\end{align}
\end{lema}
This system reduces to (\ref{monogenicas_izquierda}) when $f$ is
constant.  Thus it is natural to wish to generalize results concerning
monogenic functions to solutions of the main Vekua equation.  This is
one of our main goals in this section.

Suppose that $W\in C^2(\Omega,\H)$.  From the second equation of
(\ref{system_1}) we obtain by applying div, curl that
\begin{align} 
  \nabla \cdot f^2\,\nabla \left(\frac{W_0}{f}\right) &= 0,  \label{eq:conductivity}\\
  \curl (f^{-2} \curl (f\vec W )) &= 0. \label{eq:doublerot} 
\end{align}
The first equation is the so-called \textit{conductivity equation} and
the second one is called the \textit{double curl-type equation} for the
conductivity $f^2$.  These equations are satisfied separately by the
scalar and vector parts of $W$ in analogously to the way that two
harmonic conjugates satisfy separately the Laplace equation; together
they are not sufficient for $W_0+\vec W$ to satisfy (\ref{system_1}).
The conductivity equation is equivalent to the Schr\"odinger equation
\[ \Delta W_0 - \frac{\Delta f}{f}W_0 = 0. \]

Using (\ref{formula_Dirac}) and the fact that $f\vec W$ is vectorial,
we have the equivalence
\begin{align}\label{system_3}
  DW = \frac{Df}{f}\overline{W}  \ \iff \
   D(f\vec W) = -f^2\nabla{\left(\frac{W_0}{f}\right)}.
\end{align}
 
For brevity we will say that $f^2$ is a \textit{conductivity} when $f$
is a non-vanishing $\R$-valued function in the domain under
consideration. The conductivity will be called \textit{proper} when
$f$ and $1/f$ are bounded.
 
\subsection{Completion of Vekua solutions from partial data}

It is important to know what type of functions can be solutions to
some main Vekua equation (i.e., for some $f$).  Another question is
how to complete an $f^2$-hyperconjugate pair, i.e.\ to recover the
vector part $\vec W$ such that $W=W_0+\vec W\in\M_f(\Omega)$ when the
scalar part $W_0\colon\Omega\to \R$ is known, or vice versa.  We
apply the results of Section \ref{sec:div-curl} to these questions.
First we treat the generalization of Proposition
\ref{completacion_monogenicas} for nonconstant conductivity.

\begin{teorema}\label{teorema_extension}
  Let $f^2$ be a conductivity of class $C^2$ in an open
  star-shaped set $\Omega\subseteq\R^3$. Suppose that
  $W_0\in C^2(\Omega,\R)$ satisfies the conductivity equation
  (\ref{eq:conductivity}) in $\Omega$. Then there exists a function $\vec W$
  such that $W_0+\vec W\in\M_f(\Omega)$.  The function $f\vec W$ is
  unique up to a purely vectorial additive monogenic constant, i.e.,
  the gradient of a real harmonic function.
\end{teorema}

\begin{proof}
  Observe that (\ref{system_1}) is a homogeneous div-curl system
  (\ref{div_rot_system}) in the unknown $\vec w=f\vec W$, with $g_0=0$
  and $\vec g = -f^2\nabla(W_0/f)\in\Sol(\Omega,\R^3)$, since by
  hypothesis $W_0$ satisfies (\ref{eq:conductivity}). By Corollary
  \ref{corollary_inverse_rot}, the general solution $\vec W$ is given
  by
\begin{align}\label{completacion_vectorial}
  f\vec W = \Tiii \left[-f^2\nabla\left(\frac{W_0}{f}\right)\right] +
  \vec S_{\Omega}\left[\Ti \left[f^2\nabla\left(\frac{W_0}{f}\right)\right]\right]
  +\nabla h,
\end{align}
where $h$ is an arbitrary harmonic function. By Lemma \ref{lem:system_1},
$W_0+\vec W\in\M_f(\Omega)$.
\end{proof}

Of course, the assumptions of Theorem \ref{teorema_extension} can be
relaxed to say that $f,W_0\in H^1(\Omega,\R)$ and
(\ref{eq:conductivity}) is satisfied weakly. Then $\vec W\in H^1(\Omega,\R^3)$
given by (\ref{completacion_vectorial}) produces a weak solution
$W_0+\vec W$ of the main Vekua equation.

Analogously to the well-known $\overline{\partial}$-problem in the
complex case, the completion of the vector part of solutions of the
main Vekua equation is given in terms of the integral operator
$T_{\Omega}$. In \cite{PorShapVas1993} there is a generalization for the
quaternionic case; however, the solution given there is not
purely vectorial.

\subsection{A result on Vekua-type operators\label{sec:3dpseudo}}

As was noted at the beginning of this section, there are other
operators quite similar to the Vekua operator which appear in
factorizations of second order operators.
We write $D_r[f]=fD$ for the right-sided operator of
(\ref{formula_Dirac}). In \cite{KravTremb2011}, certain relations
were established among the operators
\begin{align*}
  \V    = D-\frac{Df}{f}C_{\H}, \quad &
  \overline{\V} = D_r-M^{\frac{Df}{f}}C_{\H}, \\
  \V_1   = D_r+\frac{Df}{f}, \quad &
  \overline{\V}_1  = D+M^{\frac{Df}{f}} .
\end{align*}
where $M^{Df/f}$ denotes right multiplication. The following is a
a right inverse of the operator $\overline{\V}$ on a subspace
analogous to the condition that $\vec g\in\Sol(\Omega,\R^3)$ for
(\ref{div_rot_system}).
 
\begin{teorema}\label{antiderivative}
  Let $f^2$ be a conductivity in the star-shaped open set
  $\Omega\subseteq\R^3$. Let $\vec G\in C^1(\Omega,\R^3)$ be a purely
  vectorial solution of the equation $\overline{\V}_1\vec G=0$. Then
  the general solution of the system $\V W=0$ and
  $\overline{\V}W=\vec G$ in $\Omega$ is given by
\begin{align*}
  W = \frac{1}{2} \left(f\A\left[\frac{\vec G}{f}\right] -
  \frac{1}{f} \Tiii[f\vec G] +
  \frac{1}{f}\vec S_{\Omega}[\Ti[f\vec G]] +
  \frac{\nabla h}{f}\right),
\end{align*}
where $h\in\Har(\Omega,\R)$ is arbitrary.
\end{teorema}

A similar formula was given in \cite[Theorem 9]{KravTremb2011} but
with an incorrect expression in place of $\Tiii-\vec S_{\Omega}\Ti$
for the inverse of curl. Otherwise the proof is essentially the same.

\subsection{ Vekua boundary value problems\label{subsec:solvekcond}}

The following fact is essential to the solution of the Calder\'on
problem in the plane \cite{AP2006}; see  \cite[Th.\ 4.1]{Isakov1998}
and the references therein, and a sketch of a proof in
$\R^n$ in \cite[p.\ 407]{AIM2009}.  The following conductivity problem reduces 
to the Dirichlet problem in the case where $f$ is constant.

\begin{proposicion}\label{ec_conductividad}
  Let $\Omega$ be a bounded domain of $\R^n$ with connected
  complement and $f^2$ a measurable proper conductivity in
  $\Omega$. Given prescribed boundary values
  $\varphi\in H^{1/2}(\partial\Omega,\R)$, there is a unique solution 
    $u\in H^1(\Omega,\R)$ to the conductivity boundary value problem
 \begin{align*}
     \nabla \cdot f^2\,\nabla u  &= 0, \\
   u|_{\partial\Omega} &= \varphi .
\end{align*}
\end{proposicion}
 
The methods of variational calculus applied in Section
\ref{subsec:variational}  can be used to obtain the
existence of solutions of second-order
elliptic equations such as this one.

 \begin{teorema}\label{mainVekuacompletion}
		Let $f^2$ be a proper conductivity in the bounded, star-shaped
    open set $\Omega\subseteq\R^3$, $f\in H^1(\Omega,\R)$ and
    suppose that $\varphi\in H^{1/2}(\partial\Omega,\R)$.  Then there
    exists a function $W\colon\Omega\to\H$ that satisfies the main
    Vekua equation (\ref{ecuacion_Vekua}) weakly and has boundary values
    $\Sc W|_{\partial\Omega}=\varphi$.
\end{teorema}

\begin{proof} Proposition \ref{ec_conductividad} gives a solution
  $u\in H^1(\Omega,\R)$ of $\nabla\cdot f^2\,\nabla u=0$ with boundary
  values
  $u|_{\partial\Omega}=\varphi/f\in H^{1/2}(\partial\Omega,\R)$.  The
  function $W_0=fu$ satisfies the conditions of Theorem
  \ref{teorema_extension} and therefore has a completion $W_0+\vec W$
  satisfying the Vekua equation weakly.
 \end{proof}

\begin{remark}
  Theorem \ref{mainVekuacompletion} provides a way to define a ``Hilbert
  transform''    
   \[ \mathcal{H}_f \colon H^{1/2}(\partial\Omega, \R) \to H^{1/2}(\partial\Omega, \R^3)  \]
  associated to the main Vekua equation  (\ref{ecuacion_Vekua}), by
  \[\mathcal{H}_f[\varphi] = \vec W|_{\partial\Omega}, \]
where $\vec W$ is given by (\ref{completacion_vectorial}).
\end{remark} 

We now characterize the elements of the space $\Vec\M_f(\Omega)$ of
vector parts of solutions to the main Vekua equation, that is the
$f^2$-hyperconjugate pair.
  
\begin{proposicion}\cite[Th.\ 10]{KravTremb2011}
  Let $\vec W\in C^2(\Omega, \R^3)$ where $\Omega$ is a simply
  connected domain in $\R^3$. For the existence of $W\in\M_f(\Omega)$
  such that $\Vec W=\vec W$ it is necessary and sufficient that
  $\div(f\vec W) = 0$ together with the double curl-type equation
  (\ref{eq:doublerot}). 
\end{proposicion}

\begin{proof}
  The necessity is given by (\ref{system_1}). For the sufficiency, the
  second condition implies that $f^{-2}\,\curl(f\vec W)$ admits a
  potential $W_0$ obtained by applying $\A$ of
  (\ref{antigradient}). The function $W=W_0+\vec W$ then satisfies
  (\ref{system_1}) and hence also (\ref{ecuacion_Vekua}).
\end{proof}

\section{Equation of double curl type \label{sec:doublerot}}

The following system of equations corresponds to the static Maxwell
system, in a medium when just the permeability $f^2$ is variable
(\cite[Ch.\ 4]{Krav2003} or \cite[Ch.\ 2]{Boss1998}):
\begin{align}   
	\div (f^2\vec{H}) & = 0,   \nonumber\\
  \div\vec E &= 0,             \nonumber\\ 
	\curl\vec H &= \vec g ,     \nonumber\\
  \curl\vec E &= f^2 \vec H.  \label{system_4}
\end{align}
Here $\vec E$ and $\vec H$ represent electric and magnetic fields,
respectively. We will apply our
results to this system and to the double curl-type equation 
\begin{align}\label{double}
  \curl (f^{-2}\,\curl\vec E) = \vec g ,
\end{align}
which is immediate from the last two equations of (\ref{system_4}).
 
\subsection{Generalized solutions of the Maxwell system
\label{subsec:generalized}}

To obtain a general solution of (\ref{system_4}) we will use the
existence of solutions of the inhomogeneous conductivity  problem
\begin{align} 
 \div(f^2\nabla W_0) &= g_0, \nonumber \\
  {W_0}|_{\partial\Omega}&=\varphi. \label{eq_conductivity_nozero}
\end{align}

\begin{teorema}[{\cite[p.\ 197, Th. 10]{Mikh1978}}]
  \label{theorem_conductivity_nozero}  
  Suppose that $f^2$ is a continuous proper conductivity in $\Omega$
  and $g_0\in L^2(\Omega,\R)$. Let
  $\varphi\in H^{1/2}(\partial\Omega,\R)$. Then there exists a unique
  generalized solution $W_0\in H^1(\Omega,\R)$ to the boundary value
  problem (\ref{eq_conductivity_nozero}). Furthermore, $W_0$ satisfies
\begin{align*}
  \|W_0\|_{H^1(\Omega)} \leq c(\|g_0\|_{L^2(\Omega)} +
  \inf\{ \|v\|_{H^1(\Omega)}\colon\ v\in H^1(\Omega,\R) \text{ and }
    v|_{\partial\Omega} = \varphi \}) 
\end{align*}
for some constant $c$ which does not depend on $g_0$ or $\varphi$.
\end{teorema}

\begin{corolario}[{\cite[p.\ 173, Th. 1]{Mikh1978} }]
  \label{corollary_conductivity_nozero} 
  Suppose that $f^2$ is a continuous proper conductivity in $\Omega$
  and $g_0\in L^2(\Omega,\R)$.  Then there exists a unique
  generalized solution $W_0\in H_0^1(\Omega,\R)$ to the boundary value
  problem (\ref{eq_conductivity_nozero}). Furthermore,
\begin{align*}
   \|W_0\|_{H_0^1(\Omega)}\leq c \|g_0\|_{L^2(\Omega)},
\end{align*}
for some constant $c$ which does not depend on $g_0$.
\end{corolario}

The right inverse of the curl given by Theorem
\ref{theorem_inverse_rot} permits us to invert the composed operator
$\curl f^{-2}\,\curl$, providing of course that this right inverse is
applied to weakly solenoidal fields.  The pair of fields
$(\vec E,\vec H)$ in the following result is constructed explicitly in
terms of the operators defined in this paper.

\begin{teorema}\label{theorem_Max_1}
  Let the domain $\Omega\subseteq\R^3$ be a star-shaped open set, and
  assume that $f^2$ is a continuous proper conductivity in $\Omega$.
  Let $\vec g \in L^2(\Omega,\R^3)$ satisfy $\div\vec g=0$. Then there
  exists a generalized solution $(\vec E,\vec H)$ to the  
  system (\ref{system_4}) and its general form is given by
	\begin{align}\label{solution_double_curl}
   \vec E &= \Tiii[f^2(\vec B + \nabla h)] -
   \vec S_{\Omega}[\Ti[f^2(\vec B + \nabla h)]]+\nabla h_1,\nonumber \\
	 \vec H &= \vec B + \nabla h,
  \end{align}
  where $h_1$ is an arbitrary real valued harmonic function.
\end{teorema}
 
\begin{proof}
  Since $\div\vec g=0$, by Corollary \ref{corollary_inverse_rot} the
  vector field
\begin{align*}
  \vec B = \Tiii[\vec g ] - \vec S_{\Omega}[\Ti[\vec g ]] 
\end{align*}
satisfies $\curl\vec B=\vec g$ and $\div\vec B=0$ weakly. To solve
\begin{align}\label{rotE}
  \curl \vec E = f^2(\vec B + \nabla h),
\end{align}
we must find an $\R-$valued function $h$  such that
\begin{align*}
 \div (f^2(\vec B + \nabla h))=0.
\end{align*}
Since $\div(f^2\vec B)=\nabla f^2 \cdot \vec B$, we need to solve
the inhomogeneous conductivity equation
\begin{align}  \label{inho_cond_eq}
 \div (f^2\nabla h)  = -\nabla f^2 \cdot \vec B,
\end{align}
It is no loss of generality to take the boundary condition
$h|_{\partial\Omega}=0$ in (\ref{eq:conductivity}). By Corollary
\ref{corollary_conductivity_nozero}, this determines a unique
generalized solution of (\ref{inho_cond_eq}) provided that
$\nabla f^2 \cdot \vec B\in L^2(\Omega,\R)$. But since
$T_{\Omega}\colon L^2(\Omega,\H)\to H^1(\Omega,\H)$ is bounded
\cite[Theorem 8.4]{GuHaSpr2008}, in fact
$\Ti[\vec g ]\in H^1(\Omega,\R)$ and
$\Tiii[\vec g ]\in H^1(\Omega,\R^3)$. Combining with the fact that
$\vec S_{\Omega}[\Ti[\vec g]]$ is harmonic by Proposition
\ref{completacion_monogenicas}, we have
$\vec B\in L^2(\Omega,\R^3)$. Thus, the hypothesis is fulfilled, and
the desired $h$ exists. Applying the right inverse of curl to
(\ref{rotE}) we have the solution (\ref{solution_double_curl})
where $h_1$ is an arbitrary harmonic function.
Then $\div\vec E=0$ by Corollary \ref{corollary_inverse_rot}, and
the remaining equations of (\ref{system_4}) are easily verified.
\end{proof}

\subsection{Variational methods for double curl boundary
  value problems\label{subsec:variational}}
In this section we will prove that given
$\vec \varphi\in H^{1/2}(\partial\Omega, \R^3)$ there exists an
extension to the interior of $\Omega$ satisfying the double curl-type
equation (\ref{eq:doublerot}). Let $f\in H^{1/2}(\Omega,\R)$ be a
measurable proper conductivity. Let us define the nonlinear functional
$\varepsilon=\varepsilon_f\colon H^1(\Omega,\R^3)\to\R$ by
\begin{align}\label{operator_rot}
  \varepsilon[\vec W] = \int_{\Omega}f^{-2}\, \curl\vec W\cdot
   \curl\vec W \,dx.
\end{align}
We are interested in proving that for fixed
$\vec\varphi\in H^{1/2}(\partial\Omega,\R^3)$, there exists at least
one element that minimizes $\varepsilon$; we will use the results of the
variational calculus which can be found, for example, in \cite[Ch.\
3]{Dac1989}. Let $X$ be a reflexive Banach space and let
$I\colon X\to\R$. We say that $I$ is weakly lower semicontinuous
(w.l.s.) if $\liminf_{k\to \infty}{I(u_k)}\geq I(u)$ whenever
$u_k\to u$ weakly in $X$.  A functional $I$ is called coercive when
there exist $\alpha>0$ and $\beta\in\R$ such that
$I(u)\geq\alpha \|u\|_X+\beta$ for all $u\in X$.

\begin{proposicion}\label{teorema_inf}(\cite[Ch.\ 3, Th.\ 1.1]{Dac1989})
  Let $X$ be a reflexive Banach space and let $I\colon X\to\R$ be a
  w.l.s.\ and coercive functional. Then there exists at least one
  element $u_0\in X$ such that
\begin{align*}
  I(u_0) = \inf\left\{I(u)\colon\ u\in X\right\}.
\end{align*}
\end{proposicion}
 
\begin{corolario}\label{corolario_inf}
  Under the hypotheses of Proposition \ref{teorema_inf}, if
  $Y\subseteq X$ is a closed (in the norm of $X$) and convex subset,
  then exists $u_1\in Y$ such that
\begin{align*}
  I(u_1) = \inf \{I(u)\colon\ u\in Y \}.
\end{align*}
\end{corolario}

We apply these facts to the reflexive Banach space
$X=H^1(\Omega,\R^3)$, and the functional $I=\varepsilon$ of
(\ref{operator_rot}), with $Y\subseteq X$ defined as follows:
\begin{align*}
  Y = \{\vec W\in H^1(\Omega,\R^3)\colon\
   \vec W|_{\partial\Omega}=\vec\varphi \}.
\end{align*}

\begin{proposicion}\label{prop:hypotheses}
  $Y\subseteq X$ and $\varepsilon$ satisfy the hypothesis of Corollary
  \ref{corolario_inf}: (a) $Y$ is convex; (b) $Y$ is closed; (c)
  $\varepsilon$ is coercive; (d) $\varepsilon$ is w.l.s.
\end{proposicion}

\begin{proof}
  (a) is immediate. To prove (b), let $\{\vec W_k\}\subseteq Y$ with
  $\vec W_k\to\vec W$; that is,
  $\|\vec W_k-\vec W\|_{H^1(\Omega)}\to0$ as $k\to\infty$.  By the
  Trace Theorem in Sobolev spaces \cite{Fournier1978} we have $C>0$
  such that
\begin{align*}
  \|\vec W_k|_{\partial\Omega}-\vec W|_{\partial\Omega}\|_{H^{1/2}(\partial\Omega)}\leq
   C\|\vec W_k - \vec W\|_{H^1(\Omega)}
\end{align*}
for all $k$.  And since $\vec W_k|_{\partial\Omega}=\vec \varphi$, then
$\vec W|_{\partial\Omega}=\vec \varphi$ almost everywhere in
$\partial\Omega$.
By definition,
\begin{align*}
 \varepsilon[\vec W] =\| f^{-1}\,\curl \vec W \|_{L^2(\Omega)}^2.
\end{align*}
so we have (c). For (d), since the norm in any Banach space is is
w.l.s., we need to prove that if $\vec W_k\to\vec W$ weakly in
$H^1(\Omega,\R^3)$, then $\curl\vec W_k\to\curl\vec W$ weakly in $L^2(\Omega,\R^3)$.
But this holds because
$\partial\vec W_k/\partial x_i\to\partial\vec W/\partial x_i$ weakly
in $L^2(\Omega,\R^3)$ ($i=1,2,3$), and because the curl is a combination of
elements of $\partial/\partial x_i$.
\end{proof}

\begin{teorema}\label{theorem_double_curl}
  Let $\Omega\subseteq\R^3$ be a bounded domain with sufficiently
  smooth boundary, and let $f^2$ be a measurable proper conductivity.
  Then given the boundary values
  $\vec\varphi\in H^{1/2}(\partial\Omega,\R^3)$ there exists an
  extension $\vec W\in H^1(\Omega,\R^3)$ such that
\begin{align} 
  \curl\left(f^{-2}\,\curl\vec W\right) &=0, \nonumber  \\
  \vec W|_{\partial\Omega} &= \vec\varphi. \label{ecuacion_rotacional}
\end{align}
\end{teorema}

\begin{proof}
  By Corollary \ref{corolario_inf} and Proposition
  \ref{prop:hypotheses}, the nonlinear functional
  (\ref{operator_rot}) has a minimum $\vec W$ over
  $[\vec{\varphi}]+H_0^1(\Omega,\R^3)$.  By definition, the second
  equation of the system (\ref{ecuacion_rotacional}) holds. To prove
  the first one, from the integration by parts formula for Sobolev spaces
  
  we have that
\begin{align}\label{rotacional_debil}
   \langle \curl f^{-2}\,\curl\vec W ,\; \vec v \rangle =
  \int_\Omega f^{-2}\,\curl\vec W \cdot \curl\vec v \,dx 
\end{align}
when $\vec v\in H_0^1(\Omega,\R^3)$. The G\^{a}teaux derivative of $\varepsilon$ at
$\vec w$ in the direction $\vec v$ is
\begin{align*}
  \varepsilon^\prime_{\vec v}[\vec W] &=
     \lim_{t\to 0} \frac{\varepsilon[\vec W+t\vec v]-\varepsilon[\vec W]}{t} \\
  &= \lim_{t\to 0} \frac{1}{t} \int_\Omega 2t f^{-2}\curl\vec v\cdot
       \curl\vec W+t^2 f^{-2}\curl\vec v \cdot \curl\vec v \,dx\\
  & = 2\int_\Omega f^{-2}\curl\vec     v \cdot \curl\vec W \,dx.
\end{align*}
Since $\vec W$ is an extreme point for
$\varepsilon$, the integral vanishes, by (\ref{rotacional_debil}) the first equation of
(\ref{ecuacion_rotacional}) holds in the distributional sense.
\end{proof}

The minimum  is not unique, because
$\varepsilon[\vec W] = \varepsilon[\vec W + \grad h]$ when
$h\in H_0^1(\Omega,\R)$.

\medskip
\noindent\textit{Further applications.}
The solution  is
analogous for the inhomogeneous counterpart of the double
curl-type equation, that is, with 
\begin{align*}
   \curl (f^{-2}\,\curl\vec W) = \vec g 
\end{align*}
in place of the first equation of (\ref{ecuacion_rotacional}). In this
case the functional to minimize is
\begin{align*}
  \varepsilon[\vec W] &= \int_\Omega f^{-2}\,\curl\vec W\cdot \curl\vec W \,dx -2
                      \int_\Omega\vec g \cdot \vec W \,dx,
\end{align*}
where $\vec g \in L^2(\Omega,\R^3)$ and $\div\vec g=0$ weakly.

Similarly, we can find weak solutions for the inhomogeneous
conductivity equation $\div f^2\nabla W_0 = g_0$ (cf.\ Proposition
\ref{ec_conductividad} and \ Theorem
\ref{theorem_conductivity_nozero}). Now the functional to minimize is
\begin{align*}
  \varepsilon[W_0] &= \int_\Omega f^2\,\nabla W_0\cdot \nabla W_0 \,dx +2
                      \int_\Omega g_0  W_0 \,dx,
\end{align*}
given $g_0\in L^2(\Omega,\R)$.

Consequently, we can relax the hypotheses of Theorem
\ref{theorem_Max_1} to assume that the proper conductivity $f^2$ is
measurable. As a result, in this generality there exists a pair
$(\vec E,\vec H)$ satisfying the differential system (\ref{system_4})
in the distributional sense, and satisfying additionally the boundary
condition
\[ \vec E|_{\partial\Omega} = \vec{\varphi}\in H^{1/2}(\partial\Omega,\R^3).\]


\begin{thebibliography}{99}

\bibitem{Fournier1978} R.~A.~Adams, J.~J.~F.~Fournier.  {\em Sobolev
    spaces}.  Academic Press, New York (1978).

\bibitem{AP2006} K.~Astala, L.~P\"{a}iv\"{a}rinta.  ``Calder\'{o}n's
  inverse conductivity problem in the plane.''  {\em Annals of
    Mathematics, 163} (2006).

\bibitem{AIM2009} K.~Astala, T.~ Iwaniec, G.~Martin.  {\em Elliptic
    Partial Differential Equations and Quasiconformal Mappings in the
    Plane}.  Princeton mathematical series. Princeton University
  Press, Princeton, Oxford (2009).


\bibitem{Bergman1969} S.~Bergman.  {\em Integral operators in the
    theory of linear partial differential equations}.  Ergebnisse der
  Mathematik und ihrer Grenzgebiete, Band 23, Springer-Verlag, New
  York (1969).

\bibitem{Bers1953} L.~Bers.  {\em Theory of pseudo-analytic functions}.
  New York University (1953).

\bibitem{Boss1998} A.~Bossavit.  {\em Computational Electromagnetism}.
  Academic Press, Boston (1998).

\bibitem{BDS1982} F.\ Brackx, R.\ Delanghe, F.\ Sommen. {\em Clifford
    analysis}. Pitman Advanced Publishing Program (1982)

\bibitem{BDSH2006} F.\ Brackx, H.\ de Schepper. ``Conjugate harmonic
  functions in Euclidean space: a spherical approach.''
 {\em Comput.\ Methods Funct.\ Theory} \textbf{6}:1 (2006) 165--182.



\bibitem{CLSSS2012} F.~Colombo, M.~E.~Luna-Elizarrar\'as, I.~Sabadini,
  M.~Shapiro, D.~C.\ Struppa, ``A quaternionic treatment of the
  inhomogeneous div-rot system.'' {\em Moscow Math. J.} \textbf{12}:1
   (2012) 37--48.

\bibitem{Colton1992}
D.~Colton, R.~Kress.
{\em Inverse Acoustic and Electromagnetic Scattering Theory}.  Springer-Verlag (1992).
 	
\bibitem{Dac1989} B.~Dacorogna.  {\em Direct methods in the Calculus
    of Variations}.  Springer-Verlag (1989).

\bibitem{Fein2005} R.~Feynman, {\em The Feynman Lectures on Physics}
  (2nd ed.). Addison-Wesley (2005).

\bibitem{Forster1981} O.~Forster.  {\em Lectures on Riemann Surfaces}.
  Grad. Texts in Math. vol. 81, Springer (1981).

\bibitem{GLS2010} J.~O.~Gonz\'{a}lez-Cervantes,
  M.~E. Luna-Elizarrar\'{a}s, M.~Shapiro.  ``On the Bergman theory for
  solenoidal and irrotational vector fields, I: General theory.''
  {\em Operator Theory: Advances and Applications} \textbf{210} 
  (2010) 79--106.
 
\bibitem{Gr1998} D.~J~Griffiths, {\em Introduction to Electrodynamics}
  (3rd ed.). Prentice Hall (1998).

\bibitem{Gri2014} Yu.~M.~Grigor'ev. ``Three-dimensional Quaternionic
  Analogue of the Kolosov-Muskhelishvili Formulae.''  {\em
    Hypercomplex Analysis: New Perspectives and Applications, Trends
    in Mathematics, Birkh\"auser, Basel} (2014) 145--166.
		

\bibitem{GuSpr1990} K.~G\"{u}rlebeck, W.~Spr\"{o}\ss{}ig.  {\em
    Quaternionic Analysis and Elliptic Boundary Value Problems}.
  Birkh\"auser Verlag, Berlin (1990).
 
\bibitem{GuSpr1997} K.~G\"{u}rlebeck, W.~Spr\"{o}\ss{}ig.  {\em
    Quaternionic and Clifford Calculus for Physicists and Engineers}.
  Chichester: John Wiley \& Sons (1997).

\bibitem{GuHaSpr2008} K.~G\"{u}rlebeck, K.~Habetha,
  W.~Spr\"{o}\ss{}ig.  {\em Holomorphic Functions in the Plane and
    n-dimensional Space}.  Birkh\"{a}user (2008).



\bibitem{Isakov1998} V.~Isakov.  {\em Inverse problems for partial
    differential equations}.  Springer-Verlag (1998).

\bibitem{Jiang1998} B.~Jiang.  {\em The Least-Squares Finite Element
    Method}.  Springer-Verlag Berlin Heidelberg (1998).

\bibitem{Jackson1999} J.~D.~Jackson.  {\em Classical electrodynamics}.
  John Wiley \& Sons, Third edition (1999).


\bibitem{Korn1968} G.~A.~Korn, T.~M.~Korn, {\em Mathematical Handbook
    for Scientists and Engineers}.  Dover Publications, Inc (1968).

\bibitem{Krav2009} V.~V.~Kravchenko.  {\em Applied pseudoanalytic
    function theory}.  Frontiers in mathematics. Birkh\"{a}user, Basel
  (2009).

\bibitem{Krav2003} V.~V.~Kravchenko.  {\em Applied Quaternionic
    Analysis}.  Heldermann Verlag: Lemgo (2003).


\bibitem{KravShap1996} V.~V.~Kravchenko, M.~V.~Shapiro.  {\em Integral
    Representations for Spatial Models of Mathematical Physics}.
  Addison Wesley Longman Ltd: Harlow (1996).


\bibitem{KravTremb2011} V.~V.~Kravchenko, S.~Tremblay.  Spatial
  pseudoanalytic functions arising from the factorization of liner
  second order elliptic operators.  {\em Mathematical Methods in the
    Applied Sciences}, \textbf{34}  (2011) 1999--2010.

\bibitem{Mikh1978} V.~P.~Mikhailov.  {\em Partial differential
    equations}.  Mir Publishers (1978).

\bibitem{PorShapVas1993} R.~M.~Porter, M.V. Shapiro, and
  N.~L. Vasilevski.  ``On the analogue of the $\bar{\partial}$-problem
  in quaternionic analysis.''  {\em Kluwer Academic Publishers Group},
  Fundamental Theories of Physics \textbf{55} (1993) 167--173.

\bibitem{Shapiro1997} M.~V.~Shapiro.  ``On the conjugate harmonic
  function of M. Riesz - E. Stein - G. Weiss.''  {\em Topics in
    Complex Analysis, Differential Geometry and Mathematical Physics,
    World Scientific} (1997) 8--32.

\bibitem{Stein1971} E.~M.~Stein, G.~Weiss.  {\em Introduction to
    Fourier Analysis on Euclidean Spaces}.  Princeton Univ. Press,
  Princeton, N.J. (1971).

\bibitem{Sud1979} A.~Sudbery. ``Quaternionic analysis.''  {\em
    Math. Proc. Cambridge Phil. Soc.} \textbf{85} (1979) 99--225.


\bibitem{Vekua1959} I.~N.~Vekua.  {\em Generalized analytic
    functions}.  Moscow: Nauka (in Russian) (1959); English translation
  Oxford: Pergamon Press (1962).

\bibitem{Weyl1940} H.~Weyl. ``The method of orthogonal projection in
  potential theory.''  {\em Duke Mathematical Journal} \textbf{7} 
  (1940) 411--444.

\end{thebibliography}
\end{document}